\documentclass[12pt]{article}
\usepackage{amssymb, amsthm,amsmath,graphics,graphicx}

\newtheorem{theorem}{Theorem}

\newtheorem{remark}[theorem]{Remark}

\title{The Edwards Model for fBm Loops and Starbursts}
\author{
	Wolfgang Bock\\
	{\small Technomathematics Group}\\
	{\small	University of Kaiserslautern}\\
	{\small	P.\ O.\ Box 3049, 67653 Kaiserslautern, Germany}\\
	{\small E-Mail:bock@mathemaik.uni-kl.de}\\[.3cm]
	Torben Fattler\\
	{\small Functional Analysis and Stochastic Analysis Group}\\
	{\small	University of Kaiserslautern}\\
	{\small	P.\ O.\ Box 3049, 67653 Kaiserslautern, Germany}\\
	{\small E-Mail:bock@mathemaik.uni-kl.de}\\[.3cm]
	Ludwig Streit\\
	{\small CIMA, University of Madeira, Campus da Penteada},\\
	{\small	9020-105 Funchal, Portugal}\\
	{\small BiBoS, Universit{ä}t Bielefeld, Germany}\\
	{\small E-Mail:streit@uma.pt}}

\begin{document}
\maketitle

\begin{abstract}
We extend Varadhan's construction of the Edwards polymer model to fractional Brownian loops and fractional Brownian starbursts. We show that, as in the fBm case, the Edwards density under a renormalizaion is an integrable function for the case $Hd\leq 1$.
\end{abstract}


\section{Introduction}
Fractional Brownian motion (fBm) has attracted considerable attention in recent years. This class of processes in general lacks the martingale and Markov properties so that many standard techniques from classical stochastic analysis are not available for them. For a detailed overview we refer to the monographs \cite{oks, mishura,  nourdin} and the references therein. There one now finds specific techniques and results developed in recent years such that these processes nowadays are more and more present in applications. Among them are models in finance, see e.g.~\cite{AM06,BaeuerleDesmettre18,BenderSottinenValkeila11} and physics~\cite{Ernst12,Metzler}. In particular they can also be used as a model for the conformations of chain polymers \cite{BBCES,Cher,GOSS,Hammouda}, generalizing the classical Brownian models (see e.g. \cite{RC}, and references therein). 
In this note we use results from \cite{Istas} to show the existence of an Edwards model \cite{Edwards, GOSS} for fractional Brownian loops and starbursts. These geometrical objects can serve as models for ring polymers and so called dendrimers, see e.g.~\cite{Tomalia,Winkler}. The existence of the Edwards density as an integrable function gives rise to the analytical study of these objects as stochastic processes. We follow here closely the lines of \cite{GOSS} for the fractional Brownian motion case and prove the additional nessessary properties for loops and starbursts.

\section{Fractional Rings}
Conventionally, fBm $ B^{H}(t),\quad t\geq 0\ $ is defined on half-lines as a centered Gaussian process with 
\begin{equation*}
	E\left( \left( B^{H}(t)-B^{H}(s)\right) ^{2}\right) =\left\vert
	t-s\right\vert ^{2H}
\end{equation*}

FBm loops should be defined with parameter $t$ on a circle of length $T$,
with translationally invariant increments around that circle. Following J.
Istas \cite{Istas}, this can be constructed replacing the distance $%
D=\left\vert t-s\right\vert $ by the geodesic one:%
\begin{equation*}
	\mathbb{E}\ \left( \left( b^{H}(s)-b^{H}(t)\right) ^{2}\right) =\left( \min
	\left( \left\vert s-t\right\vert ,T-\left\vert s-t\right\vert \right)
	\right) ^{2H}=:d^{2H}(s-t),
\end{equation*}%
but with limitations: the covariance kernel so constructed is positive
definite, and hence there is a \ corresponding Gaussian process if and only
if the Hurst index $H$ is small; for $H> 1/2\,$\ this will not be the
case, for the Brownian case in particular, with $H=1/2$, see \cite{BBS}.

For $H\leq 1/2$ one defines d-dimensional fBm loops via a n-tuple $\mathbf{b}%
^{H}=\left( b_{1}^{H},...b_{n}^{H}\right) $ of independent copies of $b^{H}.$%
We note that $d^{2H}$ is concave and positive on $(0,T)$ which implies
(see p.89 of \cite{Berman74}) that $\mathbf{b}_{H\text{ }}$is locally
non-deterministic and hence, for $0<t_{1}<\ldots <t_{n}$ there is a $k>0$
such that for any vector $\mathbf{u}:=(u_2,\dots u_n) \in \mathbb{R}^{n-1}$%
\begin{equation}
	\mathbb{E}\left( \left( \sum_{i=2}^{n}u_{i}\left( \mathbf{b}^{H}\left(
	t_{i}\right) -\mathbf{b}^{H}\left( t_{i-1}\right) \right) \right)
	^{2}\right) \geq k\sum_{i=2}^{n}u_{i}^{2}\mathbb{E}\left( \left( \mathbf{b}%
	^{H}\left( t_{i}\right) -\mathbf{b}^{H}\left( t_{i-1}\right) \right)
	^{2}\right) ,\   \label{loc}
\end{equation}%
by equation (2.1) of \cite{Berman91} and in equation (3.4) of \cite{H01}.

\subsection{The Self-Intersection Local Time}

We define, first informally, the self-intersection local time of  fBm loops
as the integral%
\begin{equation*}
	L=\underset{0<s<t<T}{\int \int }dsdt\ \delta \left( \mathbf{b}^{H}(s)-%
	\mathbf{b}^{H}(t)\right) ,
\end{equation*}%
an expression which calls for a regularization of the Dirac $\delta $%
-function, such as by the heat kernel%
\begin{equation*}
	\delta _{\varepsilon }(x)\equiv\frac{1}{(2\pi \varepsilon )^{d/2}}e^{-\frac{%
			|x|^{2}}{2\varepsilon }},\quad x\in \mathbb{R}^{d},\varepsilon >0,
\end{equation*}%
For $Hd<1,$ similar to the usual fBm case, 
\begin{equation*}
	L=\lim_{\varepsilon \searrow 0} L_{\varepsilon}:=\lim_{\varepsilon \searrow 0}\underset{0<s<t<T}{\int \int }dsdt\ \delta
	_{\varepsilon }\left( \mathbf{b}^{H}(s)-\mathbf{b}^{H}(t)\right) \ 
\end{equation*}%
exists, see e.g. Theorem 1 of \cite{HuNualart}. In particular one finds in
our case 
\begin{eqnarray}
	\mathbb{E}\left( L\right)  &=&\lim_{\varepsilon \searrow 0}\frac{1}{\left(
		2\pi \right) ^{d/2}}\int_{0}^{T}dt\int_{0}^{t}ds\frac{1}{\left(
		d^{H}(t-s)+\varepsilon \right) ^{d/2}}  \label{log} \\
	&=&\frac{1}{\left( 2\pi \right) ^{d/2}}\int_{0}^{T/2}d\tau \frac{T-\tau }{%
		\tau ^{dH}}+\frac{1}{\left( 2\pi \right) ^{d/2}}\int_{T/2\ }^{T}d\tau \frac{%
		T-\tau }{\left( T-\tau \right) ^{dH}}  \notag \\
	&=&\frac{T}{\left( 2\pi \right) ^{d/2}}\int_{0}^{T/2}d\tau \frac{1}{\tau
		^{dH}},  \notag
\end{eqnarray}%
which is finite for $Hd<1$. The same holds true for $\mathbb{E}(L^{2}).$ 

\begin{theorem}
	For $H\leq 1/2$ and $Hd<1$ there exists the $L^{2}$- limit%
	\begin{equation*}
		\lim_{\varepsilon \searrow 0}L_{\varepsilon }>0.
	\end{equation*}
\end{theorem}

Hence, for any $g>0$ there exists the Edwards model, with 
\begin{equation*}
	\frac{\exp \left( -gL\right) }{\mathbb{E}\left(\exp\left( -gL\right) \right)}\in L^{1}\left(
	\nu _{H}\right)
\end{equation*}%
as a the probability density w.r.t.~the fBm measure $\nu _{H}.$ 

\subsection{The Case Hd=1}
As in the fBm case for $Hd=1$ one has to center the local time. Hence we define
\begin{equation*}
	L_{\varepsilon ,c}\equiv L_{\varepsilon }-\mathbb{E}(L_{\varepsilon }).
\end{equation*}%

\begin{theorem}
	Assume that $Hd=1$, $d\geq 2$. Then the limit%
	\begin{equation*}
		L_{c}\equiv\lim_{\varepsilon \searrow 0}L_{\varepsilon ,c}\in L^{2}\left(\nu
		_{H}\right) 
	\end{equation*}%
	exists in $L^{2}\left(\nu _{H}\right)$ and there is a positive
	constant $M$ such that for all $0\leq g\leq M$ 
	\begin{equation}
		\exp (-gL_{c})
	\end{equation}%
	is an integrable function.\newline
\end{theorem}

\noindent Hence, also in this case, we have an Edwards measure, with 
\begin{equation*}
\frac{\exp \left( -gL_c\right) }{\mathbb{E}\left(\exp\left( -gL_c\right) \right) }\in L^{1}\left(
d\nu _{H}\right)
\end{equation*}%
as probability density w.r.t.~the fBm measure $d\nu _{H}.$

\begin{proof}
	For the case $Hd=1$ singularities arise for $\tau =d\left( t-s\right)
	\gtrsim 0,$ so that the expectation of the local time diverges. The Varadhan construction requires two estimates \cite{V69, GOSS}, namely:
	\begin{equation*}
		\mathbb{E}(L_{\varepsilon })=O(\left\vert \ln \varepsilon \right\vert )
	\end{equation*}%
	and, after centering, i.e. 
	\begin{equation*}
		L_{\varepsilon ,c}\equiv L_{\varepsilon }-\mathbb{E}(L_{\varepsilon })
	\end{equation*}%
	we need to show for some $K>0$, that
	\begin{equation*}
		\mathbb{E}\left( \left( L_{\varepsilon ,c}-L_{c}\right) ^{2}\right) \leq
		K\varepsilon ^{1/2}.
	\end{equation*}%
    In this proof we elaborate these estimates for the case of loops. The first bound can be verified
	directly, see (\ref{log}). To adapt the proof in \cite{GOSS} of the second
	estimate it is useful to introduce 
	\begin{equation*}
		\Gamma _{\varepsilon }\underset{d\left( t-s\right) \geq \Delta }{=\int \int }%
		\delta _{\varepsilon }(\mathbf{b}^{H}(t)-\mathbf{b}^{H}(s))
	\end{equation*}%
	for a small positive $\Delta $. The "gap-renormalized" $\Gamma _{\varepsilon
	}$ is non-negative and finite in the limit%
	\begin{equation*}
		\Gamma \equiv\underset{}{\underset{\varepsilon \searrow 0}{\lim }}\Gamma
		_{\varepsilon }\in L^{2}\left(\nu _{H}\right) .
	\end{equation*}%
	As a result $ \exp\left(-g \Gamma \right)$
is finite for any $g>0$.
	As a consequence we only need to verify the validity of the Varadhan construction for $$\Lambda_{\varepsilon}\equiv\underset{d\left( t-s\right) <\Delta }{\int \int }\delta
	_{\varepsilon }(\mathbf{b}^{H}(t)-\mathbf{b}^{H}(s)).$$
%
	As a first step we center $\Lambda_{\varepsilon }$: 
	\begin{equation}
		\Lambda _{\varepsilon ,c}\equiv\Lambda _{\varepsilon }-\mathbb{E}(\Lambda
		_{\varepsilon }).  \label{3Eq1}
	\end{equation}%
	For positive $\varepsilon $ one computes 
	\begin{equation*}
		\mathbb{E}(\Lambda _{\varepsilon }^{2})=\frac{1}{\left( 2\pi \right) ^{d}}%
		\int_{\mathcal{T}_{\Delta }\ }dsdtds'dt'
		\left( \left( \lambda +\varepsilon \right) \left( \rho +\varepsilon \right)
		-\mu ^{2}\right) ^{-d2}
	\end{equation*}%
	- as in equation (13) of \cite{HuNualart} - where now%
	\begin{equation*}
		\mathcal{T}_{\Delta }\equiv\{(s,t,s^{\prime },t^{\prime })\in \left[ 0,T%
		\right] ^{4}:\ \left\vert t-s\right\vert <\Delta ,\left\vert t^{\prime
		}-s^{\prime }\right\vert <\Delta \},
	\end{equation*}%
	and for $\Delta \leq T/2$ with
	\begin{eqnarray}
		\lambda &\equiv&\mathbb{E}\left( \left( b^{H}(s)-b^{H}(t)\right) ^{2}\right)
		=\left\vert t-s\right\vert ^{2H}  \label{lam1} \\
		\rho  &\equiv&\mathbb{E}\left( \left( b^{H}(s')-b^{H}(t')\right) ^{2}\right) =\left\vert t^{\prime }-s^{\prime }\right\vert ^{2H} 
		\notag \\
		\mu  &\equiv&\mathbb{E}\left( \left( b^{H}(s)-b^{H}(t)\right) \left( b^{H}(s')-b^{H}(t')\right) \right)   \notag \\
		&=&\frac{1}{2}\left( d^{2H}\left( s-t^{\prime }\right) +d^{2H}\left(
		s^{\prime \ }-t\right) -d^{2H}\left( t-t^{\prime }\right) -d^{2H}\left( s-s'\right) \right) .  \label{lambda}
	\end{eqnarray}%
	Following the argument in \cite{GOSS} we have here again the estimate%
	\begin{equation}
		\mathbb{E}\left( \left( \Lambda _{\varepsilon ,c}-\Lambda _{c}\right)
		^{2}\right) \leq \frac{d}{2(2\pi )^{d}}\int_{\mathcal{T}_{\Delta }}d\tau
		\,\rho \int_{0}^{\varepsilon }dx\,\left( \frac{1}{(\delta +x\rho )^{d/2+1}}-%
		\frac{1}{\left( (\lambda +x)\rho \right) ^{d/2+1}}\right) ,  \label{inte}
	\end{equation}%
	so, following \cite{GOSS}, it is sufficient to show that also in the case of
	loops the rhs is of order $\varepsilon ^{1/2}.$ We then decompose $%
	\mathcal{T}_{\Delta }$ into two subsets, adapting the notation of \cite{GOSS}%
	$\ $%
	\begin{equation*}
		\mathcal{T}_{\Delta 1,2}\equiv\{(s,t,s^{\prime },t^{\prime })\in \mathcal{T}%
		_{\Delta }:\ \left[ s,t\right] \cap \left[ s^{\prime },t^{\prime }\right]
		\neq \varnothing \}
	\end{equation*}%
	and%
	\begin{equation*}
		\mathcal{T}_{\Delta 3}\equiv\{(s,t,s^{\prime },t^{\prime })\in \mathcal{T}%
		_{\Delta }:\ \left[ s,t\right] \cap \left[ s^{\prime },t^{\prime }\right]
		=\varnothing \}
	\end{equation*}%
	In the first subset, for any $\Delta \leq T/4,$ "geodesic" distances $d$ \
	between any pair of points are less than $T/2,\ $hence are equal to the
	ordinary ones: $d(t^{\prime }-s)=\left\vert t^{\prime }-s\right\vert $ etc.$.
	$ Hence the estimates\ given in \cite{GOSS} for the domains\ $\mathcal{T}%
	_{1,2}$ apply to the present case, and the contribution from this subdomain
	of the integral (\ref{inte}) is of order $\varepsilon ^{1/2}.$ For $\mathcal{T}_{\Delta 3}$ we assume without loss of generality $0<s<t<s^{\prime}<t^{\prime }.$ If $\ t^{\prime }-s\leq T/2$, all distances between points $%
	(s,t,s^{\prime },t^{\prime })$ are again less than $T/2$, and as above, the
	contribution from this subdomain of the integral \eqref{inte} for sufficiently
	small $\Delta $ is of order $\varepsilon ^{1/2}$ too. Likewise, with an
	exchange of variables, for the case $d(s^{\prime }-t)\leq T/2.$
	
	In the remaining case, for sufficiently small $\Delta ,$ the geodesic
	distance $b$ between the intervals $\left[ s,t\right] $ and $\left[
	s^{\prime },t^{\prime }\right] $ is large in comparison to $\Delta $  This
	corresponds to the second sub-region%
	\begin{equation*}
		t-s<\Delta \ll b,\text{ \ }t^{\prime }-s^{\prime }<\Delta \ll b
	\end{equation*}%
	of $\mathcal{T}_{3}$ considered in the proof of Proposition 1 in \cite{GOSS}%
	. Recalling that the periodic fBm also is locally nondeterministic, we
	conclude as in the corresponding proof of Lemma3.1(3) of \cite{H01} that for
	a sufficiently small $k>0$, 
	\begin{equation*}
		\lambda \rho -\mu ^{2}\geq k\lambda \rho 
	\end{equation*}%
	also here. Hence the arguments in Lemma 6 and Lemma 7 of \cite{GOSS} carry
	over to the case at hand, and in conclusion the $O(\varepsilon ^{1/2})$
	bound holds for $\mathbb{E}\left( \left( L_{\varepsilon ,c}-L_{c}\right)
	^{2}\right) .$
\end{proof}

\section{Starbursts}

A generalization of centered Gaussian random paths, such as e.g.%
\begin{equation*}
	\mathbf{X}=\left\{ \mathbf{x}_{k}(t_{k}):\quad \mathbf{x}_{k}(0)=0;\quad k=1,\dots,n;\quad 0\leq
	t_{k}\leq T_{k}\right\}, 
\end{equation*}%
branching out from a common starting point, is often called a "starburst" or
"dendrimer" in applications. For the definition of an fBm starburst $\mathbf{X}=%
\vec{\beta}^{H}$, following \cite{Istas} one will want to maintain the characteristic fractional
correlations between different branches, i.e. 
\begin{equation*}
	\mathbb{E}\ \left( \left( {\vec{\beta} }_{k}^{H}(s)-\vec{\beta}%
	_{l}^{H}(t)\right) ^{2}\right) =d_{kl}^{2H}(s,t),
\end{equation*}%
where now the geodesic distance 
\begin{equation*}
	d_{kl}(s,t)=\left\{ 
	\begin{array}{ccc}
		\left\vert s-t\right\vert  & \text{if } & k=l \\ 
		s+t & \text{if} & k\neq l%
	\end{array}%
	\right. .
\end{equation*}%
We denote the corresponding Gaussian measure by $\mu (H,n).$\\

As shown by Istas \cite{Istas} such an extension of fBm is again viable whenever the Hurst index $H$ is no larger than $1/2$. (For $H=1/2$ this produces simply an n-tuple of independent Brownian motions.)

For $k\neq l$ we set, first informally, 
\begin{equation*}
	L_{kl}=\int_{0}^{T_{k}}ds\int_{0}^{T_{l}}dt\delta \left( \mathbf{\beta }%
	_{k}^{H}(s)-\mathbf{\beta }_{l}^{H}(t)\right)
\end{equation*}%
and for $k=l$ we define the centered local times $L_{k},_{c}$ as in \cite%
{GOSS}.

Then consider 
\begin{equation*}
	L(g)\equiv\sum_{\ k}g_{k}L_{k,c}+\sum_{l<k}g_{kl}L_{kl}
\end{equation*}%
for positive $g_k$ and $g_{kl},\quad  1\leq k,l \leq n$. \\

\noindent For shorthand we write $g>0$ for $g_k>0$ and $g_{kl}>0$, for all $1\leq k,l \leq n$. 

$L_{k,c}$ by itself is well-defined for $Hd<1$\thinspace\ and controllable 
\`{a} la Varadhan for $Hd=1$ and small positive $g$; and the $L_{kl}$ are bounded. Hence, for $Hd<1$ we have as before\\

\begin{theorem}
	For $H\leq 1/2$ and $H<1/d$ there exists the $L^{2}$ limit%
	\begin{equation*}
		\lim_{\varepsilon \searrow 0}L_{\varepsilon }(g)>0.
	\end{equation*}
\end{theorem}

\noindent So for any $g>0$ there exists the Edwards model, with 
\begin{equation*}
	\frac{\exp \left( -L(g)\right) }{\mathbb{E}\left( -L(g)\right) }\in
	L^{1}\left(\mu(H,n)\right) 
\end{equation*}%
as a probability density w.r.t. the Gaussian measure measure $d\mu (H,n).$

\begin{theorem}
	\bigskip Assume that $Hd=1$, $d\geq 2$. Then the limit%
	\begin{equation*}
		L_{c}(g)\equiv\lim_{\varepsilon \searrow 0}L_{\varepsilon ,c}(g)\in L^{2}\left(
		\mu(H,n)\right) 
	\end{equation*}%
	exists and there exists a positive constant $M$ such that for all $%
	0\leq g_{k}\leq M$ 
	\begin{equation}
		\exp (-L_{c}(g))
	\end{equation}%
	is an integrable function.\newline
\end{theorem}

\begin{proof}
	For sufficiently small ${g}>0$ we have that $\exp\left(
	-g_{k} L_{k,c}\right) \epsilon L^{1}(\mu(H,n)).$ Hence we can choose $g>0$ such that $
	\exp \left( - g_{k} L_{k,c}\right) \epsilon L^{n}$ and have%
	\begin{equation*}
		\exp \left( -\sum_{k=1}^{n} g_{k}L_{k,c}\right) \in L^{1}(\mu(H,n)).
	\end{equation*}%
	The $L_{kl}$ are non-negative and bounded, hence, for arbitrary $g_{kl}\geq 0
	$%
	\begin{equation*}
		\exp (-L_{c}(g))=\exp \left(
		-\sum_{k=1}^{n}g_{k} L_{k,c}-\sum_{k>1}^{n}g_{kl}L_{kl}\right) \in
		L^{1}(\mu (H,n)).
	\end{equation*}
\end{proof}

\begin{remark}
	Based on the construction of fBm on metric trees as parameter space by 
	\cite{Istas}, vast generalizations of this last result seem possible.
\end{remark}

\section{Conclusion and Outlook}
In this note we have generalized the methods from \cite{GOSS} from fBm to fractional loops and starbursts. These loops and starbursts can serve as models for the coformations of ring polymers and dendrimers in solvents. The existence of the Edwards density is the starting point for further analytical study of these models.

\end{document}